\ifdefined\IsLLNCS
\documentclass{llncs}
\usepackage{amsfonts,amsmath,amssymb,boxedminipage,color,url,nccmath}
\else
\documentclass[11pt]{article}
\usepackage{amsfonts,amsmath,amssymb,amsthm,boxedminipage,color,url,fullpage,nccmath}
\fi

\ifdefined\IsDVI
\usepackage[hypertex,pdfstartview=FitH,colorlinks,linkcolor=blue,filecolor=blue,citecolor=blue,urlcolor=blue]{hyperref} 
\else
\usepackage[pdfstartview=FitH,colorlinks,linkcolor=blue,filecolor=blue,citecolor=blue,urlcolor=blue]{hyperref} 
\fi

\usepackage{enumerate,paralist}
\usepackage[labelfont=bf]{caption}

\usepackage{aliascnt}
\usepackage[numbers]{natbib} 
\usepackage{cleveref}
\usepackage{xspace}
\usepackage{xstring}

\usepackage{tikz}\usetikzlibrary{arrows}
\usetikzlibrary{decorations.markings}
\usepackage{subcaption}

\newcommand{\remove}[1]{}
\newcommand{\Draft}[1]{\ifdefined\IsDraft\texttt{ #1} \fi}
\newcommand{\LLNCS}[1]{\ifdefined\IsLLNCS #1 \fi}

\newcommand{\TLLNCS}[2]{\ifdefined\IsLLNCS#1\else #2 \fi}

\ifdefined\IsDraft
    \newcommand{\authnote}[2]{{\bf [{\color{red} #1's Note:} {\color{blue} #2}]}}
\else
    \newcommand{\authnote}[2]{}
\fi

\ifdefined\IsDraft
    
\else
    
\fi

\ifdefined\IsDraft
    \newcommand{\deleted}[1]{{\color{blue} ~Deleted:~{\color{red} #1}}}
\else
    \newcommand{\deleted}[1]{}
\fi

\newcommand{\sdotfill}{\textcolor[rgb]{0.8,0.8,0.8}{\dotfill}} 

\newenvironment{protocol}{\begin{proto}}{\vspace{-\topsep}\sdotfill\end{proto}}

\newcommand{\Ensuremath}[1]{\ensuremath{#1}\xspace}
\newcommand{\MathAlg}[1]{\mathsf{#1}}

\newcommand{\MathAlgX}[1]{\Ensuremath{\MathAlg{#1}}}

\newcommand{\ie}  {i.e.,\xspace}
\newcommand{\eg}  {e.g.,\xspace}

\newcommand{\abs}[1]{\left\lvert #1 \right\rvert}
\newcommand{\ceil}[1]{\left\lceil #1 \right\rceil}

\newcommand{\set}[1]{\ens{#1}}

\newcommand{\floor}[1]{\left \lfloor#1 \right \rfloor}

\newcommand{\N}{{\mathbb{N}}}

\newcommand{\zo}{\{0,1\}}

\newcommand{\zs}{{\zo^\ast}}

\newcommand{\ci} {\stackrel{\rm{c}}{\equiv}}
\newcommand{\statclose} {\stackrel{\rm{s}}{\equiv}}

\newcommand{\poly}{\operatorname{poly}}

\newcommand{\negl}{\operatorname{neg}}

\renewcommand{\cref}{\Cref}

\TLLNCS{
\newaliascnt{claiml}{theorem}
\newtheorem{claiml}[claiml]{Claim}
\aliascntresetthe{claiml}

\renewenvironment{claim}{\begin{claiml}}{\end{claiml}}
}
{
\newtheorem{theorem}{Theorem}[section]

\newaliascnt{lemma}{theorem}
\newtheorem{lemma}[lemma]{Lemma}
\aliascntresetthe{lemma}

\newaliascnt{claim}{theorem}
\newtheorem{claim}[claim]{Claim}
\aliascntresetthe{claim}

\newaliascnt{corollary}{theorem}
\newtheorem{corollary}[corollary]{Corollary}
\aliascntresetthe{corollary}

\newaliascnt{proposition}{theorem}
\newtheorem{proposition}[proposition]{Proposition}
\aliascntresetthe{proposition}

\newaliascnt{conjecture}{theorem}

\aliascntresetthe{conjecture}

\newaliascnt{definition}{theorem}
\newtheorem{definition}[definition]{Definition}
\aliascntresetthe{definition}

\newaliascnt{remark}{theorem}

\aliascntresetthe{remark}

\newaliascnt{example}{theorem}
\newtheorem{example}[example]{Example}
\aliascntresetthe{example}
}

\crefname{lemma}{Lemma}{Lemmas}
\crefname{figure}{Figure}{Figures}
\crefname{claim}{Claim}{Claims}
\crefname{corollary}{Corollary}{Corollaries}
\crefname{proposition}{Proposition}{Propositions}
\crefname{conjecture}{Conjecture}{Conjectures}
\crefname{definition}{Definition}{Definitions}
\crefname{remark}{Remark}{Remarks}
\crefname{exmaple}{Example}{Examples}

\newaliascnt{construction}{theorem}

\aliascntresetthe{construction}
\crefname{construction}{Construction}{Constructions}

\newaliascnt{fact}{theorem}

\aliascntresetthe{fact}
\crefname{fact}{Fact}{Facts}

\newaliascnt{notation}{theorem}

\aliascntresetthe{notation}
\crefname{notation}{Notation}{Notation}

\crefname{equation}{Equation}{Equations}

\newaliascnt{proto}{theorem}

\newtheorem{proto}[proto]{Protocol}

\aliascntresetthe{proto}
\crefname{proto}{protocol}{protocols}

\newaliascnt{algo}{theorem}
\newtheorem{algo}[algo]{Algorithm}
\aliascntresetthe{algo}
\crefname{algo}{algorithm}{algorithms}

\newaliascnt{expr}{theorem}
\newtheorem{expr}[expr]{Experiment}
\aliascntresetthe{expr}
\crefname{experiment}{experiment}{experiments}

\def\FullBox{$\Box$}
\def\qed{\ifmmode\qquad\FullBox\else{\unskip\nobreak\hfil
\penalty50\hskip1em\null\nobreak\hfil\FullBox
\parfillskip=0pt\finalhyphendemerits=0\endgraf}\fi}

\def\qedsketch{\ifmmode\Box\else{\unskip\nobreak\hfil
\penalty50\hskip1em\null\nobreak\hfil$\Box$
\parfillskip=0pt\finalhyphendemerits=0\endgraf}\fi}

\renewcommand{\Pr}{{\mathrm {Pr}}}
\newcommand{\pr}[1]{\Pr\left[#1\right]}

\newcommand{\Ra}{\mathsf{R}}

\newcommand{\Rc}{\Ra}

\newcommand{\Ac}{\MathAlgX{A}}
\newcommand{\Dc}{\MathAlgX{D}}

\newcommand{\Bc}{\mathsf{B}}
\newcommand{\Cc}{\mathsf{C}}

\newcommand{\ens}[1]{\left\{#1\right\}}
\newcommand{\size}[1]{\left|#1\right|}

\newcommand{\Uni}{{\mathord{\mathcal{U}}}}

\newcommand{\prob}[1]{\mathsf{\textsc{#1}}}

\newcommand{\SD}{\prob{SD}}

\newcommand{\ppt}{{\sc ppt}\xspace}

\newcommand{\is}{{i^\ast}}

\newcommand{\Tableofcontents}{
\ifdefined\IsLLNCS \else
\thispagestyle{empty}
\pagenumbering{gobble}
\clearpage
\ifdefined\IsSubmission \else
\tableofcontents
\thispagestyle{empty}
\clearpage
\fi
\pagenumbering{arabic}
\fi
}

\newcommand{\vect}[1]{{ \boldsymbol #1}}

\newcommand{\vx}{\vect{x}}

\newcommand{\vw}{\vect{w}}

\newcommand{\party}[1]{%
    \IfEqCase{#1}{%
        {1}{\Ac}
        {2}{\Bc}
        {3}{\Cc}
    }[\PackageError{\party}{Undefined option to party: #1}{}]%
}%

\newcommand{\vparty}[2]{%
    #1_{\party{#2}}
}%

\newcommand{\wparty}[1]{%
    \vparty{w}{#1}
}%

\newcommand{\Adv}{\Dc}
\newcommand{\ptp}{{point-to-point}\xspace}
\newcommand{\secParam}{\kappa}

\newcommand{\Party}{\MathAlgX{P}}
\newcommand{\Pone}{{\party{1}}}
\newcommand{\Ptwo}{{\party{2}}}
\newcommand{\Pthree}{{\party{3}}}

\newcommand{\HalfP}{\frac n2}
\newcommand{\ThirdP}{\frac n3}

\newcommand{\CorSet}{{\mathcal{I}}}
\newcommand{\Sim}{\MathAlgX{S}}
\newcommand{\aux}{z}
\newcommand{\ys}{y^\ast}

\newcommand{\SMbox}[1]{\mbox{\scriptsize {\sc #1}}}
\newcommand{\REAL}{\SMbox{REAL}}
\newcommand{\IDEAL}{\SMbox{IDEAL}}

\newcommand{\abort}{\MathAlgX{abort}}

\newcommand{\bigbrack}{{\vphantom{2^{2^2}}}}
\mathchardef\mhyphen="2D

\ifdefined \IsLLNCS
\newcommand{\titlename}{Characterization of Secure Multiparty Computation Without Broadcast}
\else
\newcommand{\titlename}{Characterization of Secure Multiparty Computation\\ Without Broadcast}
\fi

\title{\titlename\footnote{This is the final draft of this paper. The full version was published in the Journal of Cryptology 2018 \cite{CHOR18}. An extended abstract of this work appeared in the Theory of Cryptography Conference (TCC) 2016-A \cite{CHOR16}.}
\Draft{\\{\small \sc Working Draft: Please Do Not Distribute}}
}

\ifdefined\IsAnon
    \author{}
    \date{}
    \LLNCS{ \institute{}}
\else

    \ifdefined\IsLLNCS
        \author{Ran Cohen\inst{1}\thanks{Work supported by {\sc the israel science foundation} (grant No.~189/11), the Ministry of Science, Technology and Space and by the National Cyber Bureau of Israel.}
        \and Iftach Haitner\inst{2}\thanks{Research supported by ERC starting grant 638121, ISF grant 1076/11, I-CORE grant 4/11, BSF grant 2010196, and Check Point Institute for Information Security.}
        \and Eran Omri\inst{3}\thanks{Research supported by ISF grant 544/13.}
        \and Lior Rotem\inst{4}}
        \institute{Department of Computer Science, Bar-Ilan University\\ \email{cohenrb@cs.biu.ac.il}
        \and School of Computer Science, Tel Aviv University\\ \email{iftachh@cs.tau.ac.il}
        \and Department of Computer Science and Mathematics, Ariel University\\ \email{omrier@ariel.ac.il}
        \and School of Computer Science, Tel Aviv University\\ \email{lior.rotem@gmail.com}}
    \else
        \author{Ran Cohen\thanks{Department of Computer Science, Bar-Ilan University. E-mail: \texttt{cohenrb@cs.biu.ac.il}. Work supported by {\sc the israel science foundation} (grant No.~189/11), the Ministry of Science, Technology and Space and by the National Cyber Bureau of Israel.}
        \and Iftach Haitner\thanks{School of Computer Science, Tel Aviv University. E-mail: \texttt{iftachh@cs.tau.ac.il}. Research supported by ERC starting grant 638121, ISF grant 1076/11, I-CORE grant 4/11, BSF grant 2010196, and Check Point Institute for Information Security.}
        \and Eran Omri\thanks{Department of Computer Science and Mathematics, Ariel University. E-mail: \texttt{omrier@ariel.ac.il}. Research supported by ISF grant 544/13.}
        \and Lior Rotem\thanks{School of Computer Science, Tel Aviv University. E-mail: \texttt{lior.rotem@gmail.com}.} }
    \fi
\fi

\begin{document}
\sloppy
\maketitle
\begin{abstract}
  A major challenge in the study of cryptography is characterizing the necessary and sufficient assumptions required to carry out a given cryptographic task. The focus of this work is the necessity of a broadcast channel for securely computing symmetric functionalities (where all the parties receive the same output) when one third of the parties, or more, might be corrupted. Assuming all parties are connected via a peer-to-peer network, but no broadcast channel (nor a secure setup phase) is available, we prove the following characterization:

  \begin{itemize}
    \item
    A symmetric $n$-party functionality can be securely computed facing $n/3\le t<n/2$ corruptions (\ie honest majority), if and only if it is \emph{$(n-2t)$-dominated}; a functionality is $k$-dominated, if \emph{any} $k$-size subset of its input variables can be set to \emph{determine} its output.
  	
  	\item
    Assuming the existence of one-way functions, a symmetric $n$-party functionality can be securely computed facing $t\ge n/2$ corruptions (\ie no honest majority), if and only if it is $1$-dominated and can be securely computed with broadcast.
  \end{itemize}	

  It follows that, in case a third of the parties might be corrupted, broadcast is necessary for securely computing non-dominated functionalities (in which ``small'' subsets of the inputs cannot determine the output), including, as interesting special cases, the Boolean XOR and coin-flipping functionalities.
\end{abstract}


\Tableofcontents

\section{Introduction}
Broadcast (introduced by \citet{LamportSP82} as the Byzantine Generals problem) allows any party to deliver a message of its choice to all parties, such that all honest parties will receive the same message even if the broadcasting party is corrupted. Broadcast is an important resource for implementing secure multiparty computation. Indeed, much can be achieved when broadcast is available (hereafter, the broadcast model); in the computational setting, assuming the existence of oblivious transfer, every efficient functionality can be securely computed \emph{with abort},\footnote{An efficient attack in the real world is computationally indistinguishable, via a simulator, from an attack on an ``ideal computation'', in which malicious parties are allowed to prematurely abort.} facing an arbitrary number of corruptions \cite{Yao82,GoldreichMW87}. Some functionalities can be computed with \emph{full security},\footnote{The malicious parties in the ``ideal computation'' are \emph{not} allowed to prematurely abort.} \eg Boolean OR and three-party majority \cite{GordonK09}, or $1/p$-security,\footnote{The real model is $1/p$-indistinguishable from an ``ideal computation'' without abort.} \eg coin-flipping protocols \cite{MoranNS09,HaitnerT14}. In the information-theoretic setting, considering ideally-secure communication lines between the parties, every efficient functionality can be computed with full security against unbounded adversaries,\footnote{The real and ideal models are statistically close: indistinguishable even in the eyes of an all-powerful distinguisher.} facing any minority of corrupted parties \cite{RB89}.

The above drastically changes when broadcast or a secure setup phase are not available.\footnote{In case a secure setup phase is available, \emph{authenticated broadcast} can be computed facing $t<n$ corrupted parties; Authenticated broadcast exists in the computational setting over authenticated channels assuming one-way functions exist \cite{DolevS83} and in the information-theoretic setting over secure channels assuming a limited access to a broadcast channel in the offline phase \cite{PW92}.} Specifically, when considering multiparty protocols (involving more than two parties), in which the parties are connected only via a peer-to-peer network (hereafter, the \ptp model) and one third of the parties, or more, might be corrupted.\footnote{For two-party protocols, the broadcast model is equivalent to the \ptp model (and thus all the results mentioned in the broadcast model hold also in the \ptp model). If less than a third of the parties are corrupted, broadcast can be implemented using a protocol, and every functionality can be computed with information-theoretic security \cite{BGW88,ChaumCD88}.} Considering authenticated channels and assuming the existence of oblivious transfer, every efficient functionality can be securely computed with abort, facing an arbitrary number of corruptions \cite{FGHHS02}. In the full-security model, some important functionalities \emph{cannot} be securely computed (\eg Byzantine agreement \cite{PeaseSL80} and three-party majority \cite{CohenLindell14}), whereas other functionalities can (\eg \emph{weak} Byzantine agreement \cite{FGHHS02} and Boolean OR \cite{CohenLindell14}). The characterization of many other functionalities, however, was unknown. For instance, it was unknown whether the coin-flipping functionality or the Boolean XOR functionality can be computed with full securely, even when assuming an honest majority.

\subsection{Our Result}\label{sec:into:ourResult}
A protocol is $t$-\emph{consistent}, if in any execution of the protocol, in which at most $t$ parties are corrupted, \emph{all} honest parties output the same value. Our main technical result is the following attack on consistent protocols.

\begin{lemma}[main lemma, informal]\label{lem:mainLemmaInf}
Let $n\geq 3$, $t\geq\ThirdP$ and let $s=n-2t$ if $t<\HalfP$ and $s=1$ otherwise. Let $\pi$ be an efficient $n$-party, $t$-consistent protocol in the \ptp model with secure channels. Then, there exists an efficient adversary that by corrupting any $s$-size subset $\CorSet$ of the parties can do the following: first, before the execution of $\pi$, output a value $\ys = \ys(\CorSet)$. Second, during the execution of $\pi$, force the remaining honest parties to output $\ys$.
\end{lemma}
The lemma extends to {\sf expected} polynomial-time protocols, and to protocols that only guarantee consistency to hold with high probability. We prove the lemma by extending the well-known hexagon argument of \citet{FischerLM85}, originally used for proving the impossibility of reaching (strong and weak) Byzantine agreement in the \ptp model.

A corollary of \cref{lem:mainLemmaInf} is the following lower bound on symmetric functionalities (\ie all parties receive the same output value). A functionality is \emph{$k$-dominated}, if there exists an efficiently computable value $\ys$ such that \emph{any} $k$-size subset of the functionality input variables, can be manipulated to make the output of the functionality be $\ys$ (\eg the Boolean OR functionality is $1$-dominated with value $\ys=1$).

\begin{corollary}[Informal]\label{cor:mainLInf}
Let $n\geq 3$, $t\geq\ThirdP$, and let $s=n-2t$ if $t<\HalfP$ and $s=1$ otherwise. A symmetric $n$-party functionality that can be computed with full security in the \ptp model with secure channels, facing up to $t$ corruptions, is $s$-dominated.\footnote{Stating the lower bound in the secure-channels model is stronger than stating it in the authenticated-channels model, since if a functionality can be computed with authenticated channels then it can be computed with secure channels.}
\end{corollary}
Interestingly, the above lower bound is tight. \citet{CohenLindell14} (following \citet{FGHHS02}) showed that assuming one-way functions exist, any $1$-dominated functionality (\eg Boolean OR) that can be securely computed in the broadcast model with authenticated channels, can be securely computed in the \ptp model with authenticated channels. This shows tightness when an honest majority is not assumed. We generalize the approach of \cite{CohenLindell14}, using the two-threshold detectable precomputation of \citet{FHHW03}, to get the following upper bound.
\begin{proposition}[Informal]\label{thm:mainUBInf}
Let $n\geq 3$ and $\ThirdP\leq t<\HalfP$. Assuming up to $t$ corruptions, any efficient symmetric $n$-party functionality that is $(n-2t)$-dominated can be computed in the secure-channels \ptp model with information-theoretic security.
\end{proposition}

Combining \cref{cor:mainLInf}, \cref{thm:mainUBInf} and \cite[Thm.\ 7]{CohenLindell14}, yields the following characterization of symmetric functionalities.
\begin{theorem}[main theorem, informal]\label{thm:mainCharacterizationInf}
Let $n\geq 3$, $t\geq \ThirdP$ and let $f$ be an efficient symmetric $n$-party functionality.
\begin{enumerate}
   \item
   For $t<\HalfP$, $f$ can be $t$-securely computed (with information-theoretic security) in the secure-channels \ptp model, if and only if $f$ is $(n-2t)$-dominated.
   \item
   For $t\geq\HalfP$, assuming one-way functions exist, $f$ can be $t$-securely computed (with computational security) in the authenticated-channels \ptp model, if and only if $f$ is $1$-dominated and can be $t$-securely computed (with computational security) in the authenticated-channels broadcast model.
\end{enumerate}
\end{theorem}

Another application of \cref{lem:mainLemmaInf} regards coin-flipping protocols. A coin-flipping protocol \cite{Blum81} allows the honest parties to jointly flip an unbiased coin, where even a coalition of (efficient) cheating parties cannot bias the outcome of the protocol by too much. We focus on protocols in which honest parties must output the same bit. Although \cref{thm:mainCharacterizationInf} shows that fully-secure coin flipping cannot be achieved facing one-third corruptions, we provide a stronger impossibility result under a weaker security requirement that only assumes $\ThirdP$-consistency and a non-trivial bias. In particular, we show that $1/p$-secure coin flipping cannot be achieved using consistent protocols in case a third of the parties might be corrupted.
\begin{corollary}[impossibility of many-party coin flipping in the \ptp model, informal]\label{cor:CFInf}
In the secure-channels \ptp model, there exists no $(n\geq3)$-party coin-flipping protocol that guarantees a non-trivial bias (\ie smaller than $\frac12$) against an efficient adversary controlling one third of the parties.
\end{corollary}

The above is in contrast to the broadcast model, in which coin flipping can be computed with full security if an honest majority exists \cite{BD84,CGMA85}, and $1/p$-security when no honest majority is assumed \cite{Cleve86,BeimelOO10,HaitnerT14}.

\subsection{Our Technique}\label{sec:into:technique}
We present the ideas underlying our main technical result, showing that the following holds in the \ptp model. For any efficient consistent protocol involving more than two parties, if one third of the parties (or more) might be corrupted, then there exists an efficient adversary that can make the honest parties output a predetermined value. In the following discussion we focus on three-party protocols with a single corrupted party.

Let $\pi = (\party1,\party2,\party3)$ be an efficient $1$-consistent three-party protocol, and let $q$ be its round complexity on inputs of fixed length $\secParam$. Consider the following ring network $\Rc = (\party{1}^1,\party{2}^1,\party{3}^1,\ldots,\party{1}^{q},\party{2}^{q},\party{3}^q)$, where each two consecutive parties, as well as the first and last, are connected via a secure channel, and party $\Party^{j}$, for $\Party\in \set{\Pone,\Ptwo,\Pthree}$, has the code of $\Party$ (see \cref{fig:1}).\newcommand{\Mynode}[2]{
\node[draw, shape=circle,inner sep=\inSep,minimum width=\minwidth ] (#1) at ({360/\n * (#1 -1) }:\radius) {#2};
}
 \newcommand{\Myedge}[3]{
 \draw[#3, >=latex] ({360/\n * (#2)+\margin}:\radius)
    arc ({360/\n * (#2)+\margin}:{360/\n * (#1)-\margin}:\radius);
    }
\newcommand{\Mydraw}[2]{\Mynode{#1}{#2};\Myedge{#1}{#1-1}{-};}
\def\m{q}
\def \minwidth {0.9cm}
\def\inSep{1pt}
\def\ScaleFactor{.7}
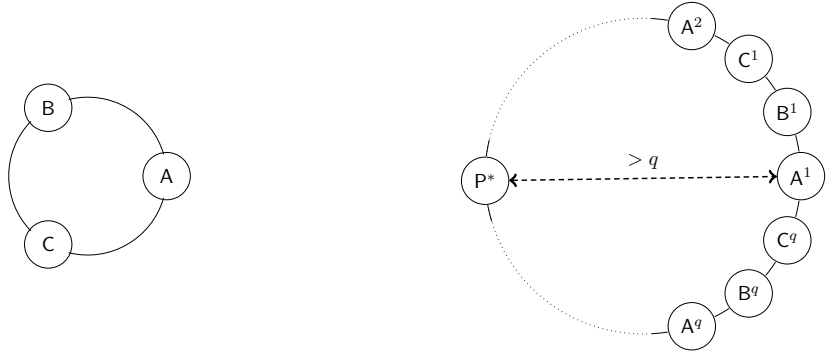
\begin{figure}
\centering
\begin{subfigure}[h]{0.35\linewidth}
\centering
\scalebox{\ScaleFactor}
{
\begin{tikzpicture}

\def \n {3}

\def \radius {1.5cm}
\def \margin {16} 

\foreach \s in {1,...,\n}
{

\Mydraw{\s}{$\party{\s}$}
}
\end{tikzpicture}
}
\end{subfigure}
\remove{
\begin{subfigure}[h]{0.25\linewidth}
\centering
\scalebox{\ScaleFactor}
{
\begin{tikzpicture}

\def \n {6}
\def \radius {1.5cm}
\def \outradius {2.2cm}
\def \margin {16} 

\foreach \s in {1,...,3}
{
\Mydraw{\s}{$\party{\s}^1$};
\node[draw=none]  at ({360/\n * (\s -1) }:\outradius) {0};

\Mydraw{\s+3}{$\party{\s}^1$};
\node[draw=none]  at ({360/\n * (\s +3 -1) }:\outradius) {1};
}

\end{tikzpicture}
}

\subcaption{\scriptsize Hexagon ---  Two copies of $\pi$ concatenated. First copy with common input $0$ and second with $1$.\label{fig:Hexagon}}
\end{subfigure}
}
\begin{subfigure}[h]{0.50\linewidth}
\centering

\scalebox{\ScaleFactor}
{
\begin{tikzpicture}
\def \n {15}
\def \radius {3cm}
\def \margin {9} 


\Myedge{12}{11}{-}

\Mydraw{13}{$\party{1}^\m$}
\Mydraw{14}{$\party{2}^\m$}
\Mydraw{15}{$\party{3}^\m$}
\Mydraw{1}{$\party{1}^1$}
\Mydraw{2}{$\party{2}^1$}
\Mydraw{3}{$\party{3}^1$}
\Mydraw{4}{$\party{1}^2$}


\def\LHSLoc{7.57}

\node[draw, shape=circle,inner sep=\inSep, minimum width=\minwidth ] (LHS) at ({360/\n * (\LHSLoc) }:\radius) {$\Party^\ast$};

\Myedge{\LHSLoc+1}{\LHSLoc}{-}
\Myedge{\LHSLoc}{\LHSLoc-1}{-}

\Myedge{11.7}{\LHSLoc}{dotted}
\Myedge{\LHSLoc}{3.25}{dotted}



\draw[thick,densely  dashed,decoration={markings, mark=at position 0.03 with {\arrow [ ultra thick] {<}},mark=at position 1 with {\arrow [ ultra thick] {>}}}, postaction={decorate}]  (1) -- node [above=1pt] {$>q$} (LHS);

\end{tikzpicture}
}
\end{subfigure}
\caption{\scriptsize The original $3$-party protocol $\pi = (\Ac,\Bc,\Cc)$ is on the left. On the right is the $3q$-Ring --- $q$ copies of $\pi$ concatenated. Communication time between parties of opposite sides is larger than $3q/2 > q$. \label{fig:1}}
\end{figure}

Consider an execution of $\Rc$ on input $\vw =(\wparty{1}^1,\wparty{2}^1,\wparty{3}^1,\ldots,\wparty{1}^q,\wparty{2}^q,\wparty{3}^q) \in (\zo^\secParam)^{3q}$ (\ie party $\Party^i$ has input $w_{\Party}^i$, containing its actual input and random coins). A key observation is that the view of party $\party{1}^{j}$, for instance, in this execution, is a \emph{valid} view of the party $\party{1}$ on input $\wparty{1}^{j}$ in an interaction of $\pi$ in which $\party{2}$ acts honestly on input $\wparty{2}^{j}$. It is also a valid view of $\party{1}$, on input $\wparty{1}^{j}$, in an interaction of $\pi$ in which $\party{3}$ acts honestly on input $\wparty{3}^{j-1 \pmod q}$. Hence, the consistency of $\pi$ yields that any two consecutive parties in $\Rc$ output the same value, and thus \emph{all} parties of $\Rc$ output the \emph{same} value.

Consider for concreteness an attack on the parties $\set{\Pone,\Ptwo}$. The efficient adversary \Adv first selects a value $\vw \in (\zo^\secParam)^{3q}$, emulates (in its head) an execution of $\Rc$ on $\vw$, and sets $\ys$ to be the output of the party $\Party^\ast = \party{1}^{q/2}$ in this execution. To interact with the parties $\set{\Pone,\Ptwo}$ in $\pi$, the adversary \Adv corrupts party $\Pthree$ and emulates an execution of $\Rc$, in which all but $\set{\party{1}^1,\party{2}^1}$ have their inputs according to $\vw$ (the roles of all parties but $\set{\party{1}^1,\party{2}^1}$ are played by the corrupted $\Pthree$), and $\set{\Pone,\Ptwo}$ take (without knowing it) the roles of $\set{\party{1}^1,\party{2}^1}$.

We claim that the output of $\set{\Pone,\Ptwo}$ under the above attack is $\ys$. Observe that the emulation of $\Rc$, induced by the interaction of \Adv with $\set{\Pone,\Ptwo}$, is just a valid execution of $\Rc$ on some input $\vw'$ (not completely known to the adversary). Hence, by the above observation, all parties in $\Rc$ (including $\set{\Pone,\Ptwo}$) output the same value at the end of this emulation. Since the execution of $\Rc$ ends after at most $q$ rounds, and since the number of communication links between $\set{\party{1}^1,\party{2}^1}$ and $\Party^\ast$ is $\approx 3q/2 > q$, the actions of $\set{\party{1}^1,\party{2}^1}$ have \emph{no effect} on the view of $\Party^\ast$. In particular, the output of $\Party^\ast$ in the attack is also $\ys$, and by the above this is also the output of $\set{\Pone,\Ptwo}$.

\paragraph{Extension to expected polynomial-time protocols.}
The above attack works perfectly if $\pi$ runs in (strict) polynomial time. For expected polynomial-time protocols, one has to work slightly harder to come up with an attack that is (almost) as good.

Let $q$ be the expected round complexity of $\pi$. That is, an honest party of $\pi$ halts after $q$ rounds in expectation, regardless of what the other parties do, where the expectation is over its random coins. Consider the ring $\Rc = (\party{1}^1,\party{2}^1,\party{3}^1,\ldots,\party{1}^{m},\party{2}^{m},\party{3}^m)$, for $m =2q$. By Markov bound, in a random execution of $\Rc$, a party halts after $m$ rounds with probability at least $\frac12$.

The adversary \Adv attacking $\set{\Pone,\Ptwo}$ is defined as follows. For choosing a value for $\ys$, it emulates an execution of $\Rc$ on arbitrary inputs and uniformly-distributed random coins. If the party $\Party^\ast = \party{1}^{m/2}$ halts in at most $m$ rounds, \Adv sets $\ys$ to be $\Party^\ast$'s output, and continues to the second stage of the attack. Otherwise, it emulates $\Rc$ on new inputs and random coins. Note that in $k$ attempts, \Adv finds a good execution with probably (at least) $1 - 2^{-k}$. After finding $\ys$, the adversary \Adv continues as in the strict polynomial case discussed above.

The key observation here is that in the emulated execution of $\Rc$, induced by the interaction of \Adv with $\set{\Pone,\Ptwo}$, the party $\Party^\ast$ \emph{never} interacts in more than $m$ communication rounds. Therefore, again, being far from $\set{\Pone,\Ptwo}$, their actions do not affect $\Party^\ast$ in the first $m$ rounds, and so do not affect it at all. Hence, $\Party^\ast$ outputs $\ys$ also in the induced execution, and so do the parties $\set{\Pone,\Ptwo}$.

\subsection{Additional Related Work}\label{sec:relatedWork}
\paragraph{Negative results.}
In their seminal work, \citet{LamportSP82} defined the problem of simulating a broadcast channel in the \ptp model in terms of the Byzantine agreement problem. They showed that a broadcast protocol exists if and only if more than two-thirds of the parties are honest. \citet{Lamport83} defined the weak Byzantine agreement problem, and showed that even this weak variant of agreement cannot be computed, using deterministic protocols, facing one-third corruptions. \citet{FischerLM85} presented simpler proofs to the above impossibility results using the so-called hexagon argument, which is also the basis of our lower bound (see \cref{sec:into:technique}). They assumed a protocol exists for the three-party case, and composed multiple copies of this protocol into a ring system that contains an internal conflict. Since the ring system cannot exist, it follows that the three-party protocol does not exist. We remark that the result of \cite{FischerLM85} extends to public-coins protocols, where parties have access to a common random string. It follows that coin flipping is not sufficient for solving Byzantine agreement, and thus the impossibility result for coin flipping stated in \cref{cor:CFInf} is not implied by the aforementioned impossibility of Byzantine agreement.

\citet{CohenLindell14} analyzed the relation between security in the broadcast model and security in the \ptp model, and showed that some (non $1$-dominated) functionalities, \eg three-party majority, that can be computed in the broadcast model cannot be securely computed in the \ptp model, since they imply the existence of broadcast.

\paragraph{Positive results.}
If the model is augmented with a trusted setup phase, \eg a public-key infrastructure (PKI), then Byzantine agreement can be computed facing any number of corrupted parties \cite{LamportSP82}. \citet{PW92} presented an information-theoretic broadcast protocol assuming a temporary broadcast channel is available during the setup phase. \citet{FGHHS02} presented a probabilistic protocol that securely computes weak Byzantine agreement facing an arbitrary number of corrupted parties. \citet{CohenLindell14} showed (using the protocol from \cite{FGHHS02}) that assuming the existence of one-way functions, any $1$-dominated functionality that can be securely computed in the broadcast model, can also be securely computed in the \ptp model.

\citet{GoldwasserL02} presented a weaker definition for MPC without agreement, in which non-unanimous abort is permitted, \ie some of the honest parties may receive output while other honest parties might abort. Using this weaker definition, they utilized non-consistent protocols and constructed secure protocols in the \ptp model, assuming an arbitrary number of corrupted parties.

\subsection{Open Questions}
Our result for the non honest-majority case (second item of \cref{thm:mainCharacterizationInf}), requires the existence of one-way functions. In particular, given a protocol $\pi$ for computing a $1$-dominated functionality $f$ with full security in the broadcast model, one-way functions are used for compiling $\pi$ into a protocol for computing $f$ with full security in the \ptp model.\footnote{For some trivial functionalities, \eg constant functions, there exist information-theoretically secure protocols in the \ptp model that are not based on such a compilation, and this extra assumption is not needed.} It might be, however, that the existence of such a broadcast-model protocol (for non-trivial functionalities) implies the existence of one-way functions, and thus adding this extra assumption is not needed.

A different interesting challenge is characterizing which \emph{non}-symmetric functionalities can be computed in the \ptp model, in the spirit of what we do here for symmetric functionalities. For example, can a three-party coin flipping in which only two parties learn the outcome coin, be computed with full security facing a single corruption?

\subsection*{Paper Organization}
Basic definitions can be found in \cref{sec:Preliminaries}. Our attack is described in \cref{sec:Attack}, and its applications are given in \cref{sec:Application}. The characterization is presented in \cref{sec:Characterization}.

\section{Preliminaries}\label{sec:Preliminaries}
\subsection{Notations}\label{sec:notations}
We use calligraphic letters to denote sets, uppercase for random variables, lowercase for values, boldface for vectors, and sans-serif (\eg \Ac) for algorithms (\ie Turing Machines).
For $n\in\N$, let $[n]=\set{1,\cdots,n}$. Let $\poly$ denote the set all positive polynomials and let \ppt denote a probabilistic algorithm that runs in \emph{strictly} polynomial time. A function $\nu \colon \N \mapsto [0,1]$ is \textit{negligible}, denoted $\nu(\secParam) = \negl(\secParam)$, if $\nu(\secParam)<1/p(\secParam)$ for every $p\in\poly$ and large enough $\secParam$.

The statistical distance between two random variables $X$ and $Y$over a finite set $\Uni$, denoted $\SD(X,Y)$, is defined as $\frac12 \cdot \sum_{u\in \Uni}\size{\pr{X = u}- \pr{Y = u}}$. We say that $X$ and $Y$ are $\delta$-close if $\SD(X,Y)\le \delta$ and statistically close (denoted $X\statclose Y$) is they are $\delta$-close and $\delta$ is negligible.

Two distribution ensembles $X=\set{X(a,\secParam)}_{a\in\zs,\secParam\in\N}$ and $Y=\set{Y(a,\secParam)}_{a\in\zs,\secParam\in\N}$ are computationally indistinguishable (denoted $X\ci Y$) if for every non-uniform polynomial-time distinguisher $\Adv$ there exists a function $\nu(\secParam) = \negl(\secParam)$, such that for every $a\in\zs$ and all sufficiently large $\secParam$'s
$$
\abs{\pr{\Adv(X(a,\secParam),1^\secParam)=1} - \pr{\Adv(Y(a,\secParam),1^\secParam)=1}}\leq \nu(\secParam).
$$

\subsection{Protocols}\label{sec:Protocols}
An $n$-party protocol $\pi= (\Party_1,\ldots,\Party_n)$ is an $n$-tuple of probabilistic interactive TMs. The term \emph{party $\Party_i$} refers to the $i$'th interactive TM. Each party $\Party_i$ starts with input $x_i\in\zs$ and random coins $r_i\in\zs$. Without loss of generality, the input length of each party is assumed to be the security parameter $\secParam$. An \emph{adversary} \Adv is another interactive TM describing the behavior of the corrupted parties. It starts the execution with input that contains the identities of the corrupted parties and their private inputs, and possibly an additional auxiliary input.
The parties execute the protocol in a synchronous network. That is, the execution proceeds in rounds: each round consists of a \emph{send phase} (where parties send their message from this round) followed by a \emph{receive phase} (where they receive messages from other parties).

In the \emph{point-to-point (communication) model}, which is the one we assume by default, all parties are connected via a \emph{fully-connected point-to-point network}. We consider two models for the communication lines between the parties: In the \emph{authenticated-channels} model, the communication lines are assumed to be ideally authenticated but not private (and thus the adversary cannot modify messages sent between two honest parties but can read them). In the \emph{secure-channels} model, the communication lines are assumed to be ideally private (and thus the adversary cannot read or modify messages sent between two honest parties). In the \emph{broadcast model}, all parties are given access to a physical broadcast channel in addition to the point-to-point network. In both models, no preprocessing phase is available.

Throughout the execution of the protocol, all the honest parties follow the instructions of the prescribed protocol, whereas the corrupted parties receive their instructions from the adversary. The adversary is considered to be \emph{malicious}, meaning that it can instruct the corrupted parties to deviate from the protocol in any arbitrary way. At the conclusion of the execution, the honest parties output their prescribed output from the protocol, the corrupted parties output nothing and the adversary outputs an (arbitrary) function of its view of the computation (containing the views of the corrupted parties). The view of a party in a given execution of the protocol consists of its input, its random coins, and the messages it sees throughout this execution.

\subsubsection{Time and Round Complexity}
We consider both \emph{strict} and \emph{expected} bounds on time and round complexity.

\begin{definition}[time complexity]\label{def:protTC}
Protocol $\pi= (\Party_1,\ldots,\Party_n)$ is a {\sf $T$-time protocol}, if for every $i \in [n]$ and every input $x_i\in\zs$, random coins $r_i\in\zs$, and sequence of messages $\Party_i$ receives during the course of the protocol, the running time of an honest party $\Party_i$ is at most $T(|x_i|)$. If $T\in \poly$, then $\pi$ is of {\sf (strict) polynomial time}.

Protocol $\pi$ has an {\sf expected running time $T$}, if for every $i \in [n]$, every input $x_i\in\zs$ and sequence of messages $\Party_i$ receives during the course of the protocol, the expected running time of an honest party $\Party_i$, over its random coins $r_i$, is at most $T(\size{x_i})$. If $T\in \poly$, then $\pi$ has {\sf expected polynomial running time}.
\end{definition}

\begin{definition}[round complexity]\label{def:protRC}
Protocol $\pi= (\Party_1,\ldots,\Party_n)$ is a {\sf $q$-round protocol}, if for every $i \in [n]$ and every input $x_i\in\zs$, random coins $r_i\in\zs$, and sequence of messages $\Party_i$ receives during the course of the protocol, the round number in which an honest party $\Party_i$ stops being active (\ie stops sending and receiving messages) is at most $q(|x_i|)$. If $q\in \poly$, then $\pi$ has {\sf (strict) polynomial round complexity}.

Protocol $\pi$ has an {\sf expected round complexity $q$}, if for every $i \in [n]$, every input $x_i\in\zs$ and sequence of messages $\Party_i$ receives during the course of the protocol, the expected round number in which an honest party $\Party_i$ stops being active, over its random coins $r_i$, is at most $q(\size{x_i})$. If $q\in \poly$, then $\pi$ has {\sf expected polynomial round complexity}.
\end{definition}

\section{Attacking Consistent Protocols}\label{sec:Attack}
In this section, we present a lower bound for secure protocols in the secure-channels \ptp model. Protocols in consideration are only assumed to have a very mild security property (discussing the more standard notion of security is deferred to \cref{sec:Application}). Specifically, we only require the protocol to be consistent -- all honest parties output the same value. We emphasize that in a consistent protocol, a party may output the special error symbol $\bot$ (\ie abort), but it can only do so if all honest parties output $\bot$ as well.

\begin{definition}[consistent protocols]\label{def:consistency}
A protocol $\pi$ is {\sf $(\delta,t)$-consistent} against $C$-class (\eg polynomial-time, expected polynomial-time) adversaries, if the following holds. Consider an execution of $\pi$ on security parameter $\secParam$, and any vector of inputs of length $\secParam$ for the parties, in which a $C$-class adversary controls at most $t$ parties. Then with probability at least $\delta(\secParam)$, all honest parties output the \emph{same} value, where the probability is taken over the random coins of the adversary and of the honest parties.
\end{definition}

\ifdefined\IsCR
We now present an attack on consistent protocols whose round complexity is strictly bounded.
An extension of the attack to consistent protocols with a bound on their \emph{expected} number of rounds appears in the full version of this paper \cite{CHOR15}.
\else
In \cref{sec:LB:strict} we present an attack on consistent protocols whose round complexity is strictly bounded, and in \cref{sec:LB:expected} we extend the attack to consistent protocols with a bound on their \emph{expected} number of rounds.

\subsection{Protocols of Strict Running-Time Guarantee}\label{sec:LB:strict}
\fi

\begin{lemma}\label{lem:AdvStrictPoly}
Let $n\ge 3$, let $t\geq \ThirdP$, and let $s=n-2t$ if $t<\HalfP$ and $s=1$ otherwise. Let $\pi$ be an $n$-party, $T$-time, $q$-round protocol in the secure-channels \ptp model that is $(1-\delta,t)$-consistent against $(T_\Adv = 2nqT)$-time adversaries. Then, there exists a $T_\Adv$-time adversary $\Adv$ such that given the control over any $s$-size subset $\CorSet$ of parties, the following holds: on security parameter $\secParam$, \Adv first outputs a value $\ys = \ys(\CorSet)$. Next, \Adv interacts with the remaining honest parties of $\pi$ on inputs of length $\secParam$, and except for probability at most $\left(\frac32 \cdot q(\secParam)+1\right)\cdot \delta(\secParam)$, the output of {\sf every} honest party in this execution is $\ys$.\footnote{We would get slightly better parameters using an attack in which at least one honest party (but not necessarily all) outputs $\ys$.}
\end{lemma}

For a polynomial-time protocol that is $(1-\negl,t)$-consistent against \ppt adversaries and assuming an honest majority, \cref{lem:AdvStrictPoly} yields a \ppt adversary that by controlling $n-2t$ of the parties can manipulate the outputs of the honest parties (\ie forcing them all to be $\ys$) with all but a negligible probability. If an honest majority is not assumed, the adversary can manipulate the outputs of the honest parties, by controlling any single party, except for a negligible probability.

We start by proving the lemma for three-party protocols, and later prove the multiparty case using a reduction to the three-party case. We actually prove a stronger statement for the three-party case, where the value $\ys$ is independent of the set of corrupted parties.
\begin{lemma}\label{lem:AdvStrictPoly_threeparty}
Let $\pi$ be a $3$-party, $q$-round protocol in the secure-channels \ptp model, let $T$ be the combined running-time of all three parties.\footnote{This is more general than $T$-time $3$-party protocols, as it captures asymmetry between the running time of the parties; this measure will turn out to be useful for proving \cref{lem:AdvStrictPoly}.} If $\pi$ is $(1-\delta,1)$-consistent against $(T_\Adv = 2qT)$-time adversaries, then there exists a $T_\Adv$-time adversary $\Adv$ such that the following holds. On security parameter $\secParam$, \Adv first outputs a value $\ys$. Next, given the control over \emph{any} non-empty set of parties, \Adv interacts with the remaining honest parties of $\pi$ on inputs of length $\secParam$, and except for probability at most $\frac32 \cdot q(\secParam)\cdot \delta(\secParam)$, the output of {\sf every} honest party in this execution is $\ys$.
\end{lemma}
\begin{proof}
We fix the input-length parameter $\secParam$ and omit it from the notation when its value is clear from the context.
Let $\pi = (\Pone,\Ptwo,\Pthree)$ and let $m=q$ (assume for ease of notation that $m$ is even). Consider, without loss of generality, that a single party is corrupted (the case of two corrupted parties follows from the proof) and assume for concreteness that the corrupted party is $\Pthree$.
Consider the following ring network $\Rc = (\party{1}^1,\party{2}^1,\party{3}^1,\ldots,\party{1}^{m},\party{2}^{m},\party{3}^m)$, in which each two consecutive parties, as well as the first and last, are connected via a secure channel, and party $\Party^{j}$, for $\Party\in \set{\Pone,\Ptwo,\Pthree}$, has the code of $\Party$. Let $v = \secParam + T(\secParam)$, and consider an execution of $\Rc$ with arbitrary inputs and uniformly-distributed random coins for the parties being $\vw = (\wparty{1}^1,\wparty{2}^1,\wparty{3}^1,\ldots,\wparty{1}^m,\wparty{2}^m,\wparty{3}^m) \in (\zo^v)^{3m}$ (\ie party $\Party^i$ has input $w_{\Party}^i$, containing its actual input and random coins).

A key observation is that the point of view of the party $\party{1}^{j}$, for instance, in such an execution, is a \emph{valid} view of the party $\party{1}$ on input $\wparty1^{j}$ in an execution of $\pi$ in which $\party{2}$ acts honestly on input $\wparty2^{j}$. It is also a valid view of $\party{1}$, on input $\wparty1^{j}$, in an execution of $\pi$ in which $\party{3}$ acts honestly on input $\wparty3^{j-1 \pmod m}$. This observation yields the following consistency property of $\Rc$.
\begin{claim}\label{claim:AdvStrictPoly}
Consider an execution of $\Rc$ on joint input $\vw\in (\zo^v)^{3m}$, where the parties' coins in $\vw$ are chosen uniformly at random, and the parties' (actual) inputs are chosen arbitrarily. Then parties of distance $d$ in $\Rc$, measured by the (minimal) number of communication links between them, as well as all $d-1$ parties between them, output the same value with probability at least $1 - d \delta$.
\end{claim}

\begin{proof}
Consider the pair of neighboring parties $\set{\Pone^j,\Ptwo^j}$ in the ring $\Rc$ (an analogous argument holds for any two neighboring parties). Let \Adv be an adversary, controlling the party $\Pthree$ of $\pi$ that interacts with $\set{\Pone,\Ptwo}$ by emulating an execution of $\Rc$ on arbitrary inputs and uniform random coins (apart from the roles of $\set{\Pone^j,\Ptwo^j}$), and let $\set{\Pone, \Ptwo}$ take (without knowing that) the roles of $\set{\Pone^j,\Ptwo^j}$ in this execution. The joint view of $\set{\Pone, \Ptwo}$ in this emulation has the same distribution as the joint view of $\set{\Pone^j,\Ptwo^j}$ in an execution of $\Rc$ with uniform random coins. Hence, the $(1-\delta)$-consistency of $\pi$ yields that $\Pone^j$ and $\Ptwo^j$ output the same value in an execution of $\Rc$ on $\vw\in (\zo^v)^{3m}$ (where the random coins within $\vw$ of each party are chosen uniformly at random) with probability at least $1-\delta$. The proof follows by a union bound.
\end{proof}

The adversary \Adv first selects a value for $\vw\in (\zo^v)^{3m}$, consisting of arbitrary input values (\eg zeros) and uniformly-distributed random coins, and sets $\ys$ to be the output of $\Party^\ast = \party{1}^{m/2}$ in the execution of $\Rc$ on $\vw$. To interact with $\set{\Pone,\Ptwo}$ in $\pi$, \Adv emulates an execution of $\Rc$ in which all but $\set{\party{1}^1,\party{2}^1}$ have their inputs according to $\vw$, and $\set{\Pone,\Ptwo}$ take the roles of $\set{\party{1}^1,\party{2}^1}$. The key observation is that the view of party $\Party^\ast$ in the emulation induced by the above attack, is the \emph{same} as its view in the execution of $\Rc$ on $\vw$ (regardless of the inputs of $\set{\Pone,\Ptwo}$). This is true since the execution of $\Rc$ ends after at most $m$ communication rounds. Thus, the actions of $\set{\Pone,\Ptwo}$ have no effect on the view of $\Party^\ast$, and therefore the output of $\Party^\ast$ is $\ys$ also in the emulated execution of $\Rc$. Finally, since all the parties in the emulated execution of $\Rc$ have uniformly-distributed random coins, and since the distance between $\Party^\ast$ and $\set{\Pone,\Ptwo}$ is (less than) $\frac{3m}2$, \cref{claim:AdvStrictPoly} yields that with probability at least $1- \frac{3m}2 \cdot \delta$, the output of $\set{\Pone,\Ptwo}$ under the above attack is $\ys$.

Note that the value $\ys$ does not depend on the identity of the corrupted party, since in the first step $\ys$ is set independently of $\Pthree$, and in the second step the attack follows without any change when the honest parties play the roles of $\set{\Ptwo^1,\Pthree^1}$ if $\Pone$ is corrupted or $\set{\Pone^2,\Pthree^1}$ if $\Ptwo$ is corrupted.
\end{proof}

We now proceed to prove \cref{lem:AdvStrictPoly} in the many-party case.
\begin{proof}
Let $\pi = (\Party_1,\ldots,\Party_n)$ be a $T$-time, $q$-round, $n$-party protocol that is $(1-\delta,t)$-consistent against $2nqT$-time adversaries. We will show an adversary that by controlling any $s$ corrupted parties, manipulates all honest parties to output a predetermine value. We separately handle the case that $\ThirdP\leq t<\HalfP$ and the case $\HalfP\leq t<n$.

\paragraph{Case $\ThirdP\leq t<\HalfP$.}
 Let $\CorSet\subseteq[n]$ be a subset of size $s= n-2t$, representing the indices of the corrupted parties in $\pi$. Consider the three-party protocol $\pi'=(\Pone',\Ptwo',\Pthree')$, defined by partitioning the set $[n]$ into three subsets $\set{\CorSet_{\Pone'},\CorSet_{\Ptwo'},\CorSet}$, where $\CorSet_{\Pone'}$ and $\CorSet_{\Ptwo'}$ are each of size $t$, and letting $\Pone'$ run the parties $\set{\Party_i}_{i\in\CorSet_{\Pone'}}$, $\Ptwo'$ run the parties $\set{\Party_i}_{i\in\CorSet_{\Ptwo'}}$ and $\Pthree'$ run the parties $\set{\Party_i}_{i\in\CorSet}$. Each of the parties in $\pi'$ waits until all the virtual parties it is running halt, arbitrarily selects one of them and outputs the virtual party's output value.

Since the subsets $\CorSet_{\Pone'}, \CorSet_{\Ptwo'}, \CorSet$ are of size at most $t$, the $q$-round, $3$-party protocol $\pi'$ is $(1-\delta,1)$-consistent against $2nqT$-adversaries (otherwise there exists a $2nqT$-time adversary against the consistency of $\pi$, corrupting at most $t$ parties). In addition, since the combined time complexity of all three parties is $nT$, by \cref{lem:AdvStrictPoly_threeparty} there exists a $2nqT$-time adversary $\Adv'$ that first determines a value $\ys$, and later, given control over any party in $\pi'$ (in particular $\Pthree'$), can force the two honest parties to output $\ys$ with probability at least $1-\frac{3q\delta}2$.

The attacker $\Adv$ for $\pi$, controlling the parties indexed by $\CorSet$, is defined as follows: In the first step, $\Adv$ runs $\Adv'$ and outputs the value $\ys$ that $\Adv'$ outputs. In the second step, $\Adv$ interacts with the honest parties in $\pi$ by simulating the parties $\set{\Pone',\Ptwo'}$ to $\Adv'$, \ie $\Adv$ runs $\Adv'$ and sends every message it receives from $\Adv'$ to the corresponding honest party in $\pi$, and similarly, whenever $\Adv$ receives a message from an honest party in $\pi$ it forwards it to $\Adv'$. It is immediate that there exists $i\in\CorSet_{\Pone'}$ such that $\Party_i$ outputs $\ys$ in the execution of $\pi$ with the same probability that $\Pone'$ outputs $\ys$ in the execution of $\pi'$, \ie with probability at least $1-\frac{3q\delta}2$. From the consistency property of $\pi$, all honest parties output the same value with probability at least $1- \delta$, and using the union bound we conclude that the output of all honest parties in $\pi$ under the above attack is $\ys$ with probability at least $1-(\frac{3q\delta}2 +\delta)$.

\paragraph{Case $\HalfP\leq t<n$.} Let $\is \in [n]$ be the index of the corrupted party in $\pi$ and consider the three-party protocol $\pi'=(\Pone',\Ptwo',\Pthree')$ defined by partitioning the set $[n]$ into three subsets $\set{\CorSet_{\Pone'},\CorSet_{\Ptwo'},\set{\is}}$, for $\size{\CorSet_{\Pone'}}=\ceil{\frac{n-1}2}$ and $\size{\CorSet_{\Ptwo'}}=\floor{\frac{n-1}2}$. As in the previous case, the size of each subset $\CorSet_{\Pone'}, \CorSet_{\Ptwo'}, \set{\is}$ is at most $t$, and the proof proceeds as above.
\end{proof}

\ifdefined\IsCR
\else
\subsection{Protocols of Expected Running-Time Guarantee}\label{sec:LB:expected}
In this section we extend the attack presented above to consistent protocols with bound on their \emph{expected} number of rounds.

\begin{lemma}\label{lem:AdvExpectedPoly}
Let $n\ge 3$, let $t\geq \ThirdP$, let $s=n-2t$ if $t<\HalfP$ and $s=1$ otherwise, and let $z$ be an integer function. Let $\pi$ be an $n$-party protocol of expected running time $T$ and expected round complexity $q$ in the secure-channels \ptp model, that is $(1-\delta,t)$-consistent against adversaries of expected running time $T_\Adv = 2n(z+1)qT$. Then, there exists an adversary $\Adv$ of expected running-time $T_\Adv$ such that given the control over any $s$-size subset $\CorSet$ of parties, the following holds: on security parameter $\secParam$, \Adv first outputs a value $\ys = \ys(\CorSet)$. Next, \Adv interacts with the remaining honest parties of $\pi$ on inputs of length $\secParam$, and except for probability at most $2\cdot \left(3\cdot q(\secParam) + 1\right) \cdot \delta(\secParam) + 2^{-z(\secParam)}$, the output of {\sf every} honest party in this execution is $\ys$.
\end{lemma}

\begin{proof}
We prove the lemma for the three-party case, the proof for the general case is similar to the proof of \cref{lem:AdvStrictPoly}. We fix the input-length parameter $\secParam$ and omit it from the notation when clear from the context.
	
Let $\pi = (\Pone,\Ptwo,\Pthree)$ and let $m=2q$. Similarly to the proof of \cref{lem:AdvStrictPoly_threeparty}, we consider the (now double size) ring $\Rc = (\party{1}^1,\party{2}^1,\party{3}^1,\ldots,\party{1}^{m},\party{2}^{m},\party{3}^m)$. The attacker \Adv follows in similar lines to those used in the proof of \cref{lem:AdvStrictPoly_threeparty}. The main difference is that in order to select $\ys$, \Adv iterates the following for $z$ times. In each iteration, \Adv emulates an execution of the ring $\Rc$ on arbitrary inputs and uniformly-distributed random coins,\footnote{Note that now we have no a priori bound on the number of random coins used by the parties. Yet, the emulation can be done in expected time $nmT$.} for $m$ communication rounds. If during one of these iterations party $\Party^\ast = \Pone^{m/2}$ halts, $\Adv$ sets $\ys$ to be its output in this iteration. Otherwise, in case the value $\ys$ was not set during all $z$ iterations, $\Adv$ outputs $\perp$ and aborts. The attack continues as in the proof of \cref{lem:AdvStrictPoly_threeparty}.
	
To analyze the above attack, we first present an upper bound on the probability that \Adv aborts.
\begin{claim}\label{claim:probAbort}
$\pr{\Adv \mbox{ aborts }} \le 2^{-z}$.
\end{claim}
\begin{proof}
By Markov bound, the probability that in a single iteration of \Adv the party $\Party^\ast$ does not halt within $m=2q$ rounds, is at most $\frac12$. Therefore, the probability that $\ys$ is not set in all $z$ iterations is at most $2^{-z}$.
\end{proof}
Since $\Party^\ast$ halts in the iteration that produced $\ys$ within $m$ rounds, its view in the emulated execution of $\Rc$ induced by the attack, is the \emph{same} as in this selected iteration (this holds even though some parties might run for more rounds in the emulated execution). In particular, $\Party^\ast$ outputs $\ys$ also in the emulated execution. The proof continues as in the proof of \cref{lem:AdvStrictPoly_threeparty}, where the only additional subtlety is that it is no longer true that the random coins of the parties in the emulated execution induced by the attack are uniformly distributed. Indeed, we have selected a value for $\vw$ that causes $\Party^\ast$ to halt within $m$ rounds. Yet, since in a random execution of $\Rc$, party $\Party^\ast$ halts within $m$ rounds with probability at least $\frac12$, the method used to sample $\vw$ at most doubles the probability of inconsistency in the ring. It follows that the attacked parties $\set{\Pone,\Ptwo}$ output $\ys$ with probability at least $1 - 2\cdot 3q\delta $ times the probability that \Adv does not abort, and the proof of the lemma follows.
\end{proof}
\fi
\section{Impossibility Results for Secure Computation}\label{sec:Application}
In this section, we present applications of the attack of \cref{sec:Attack} to secure multiparty computations in the secure-channels \ptp model.\footnote{Note that a lower bound in the secure-channels model is stronger than in the authenticated-channels model.} In \cref{sec:RealIdeal}, we show that the only symmetric functionalities that can be securely realized, according to the real/ideal paradigm, in the presence of $n/3 \le t< n/2$ corrupted parties (\ie honest majority), are $(n-2t)$-dominated functionalities. The only symmetric functionalities that can be securely realized in the presence of $n/2 \le t< n$ corrupted parties (\ie no honest majority), are $1$-dominated functionalities. In \cref{sec:CF}, we show that non-trivial $(n>3)$-party coin-flipping protocols, in which the honest parties must output a bit, are impossible when facing $t \ge n/3$ corrupted parties.

For concreteness, we focus on strict polynomial-time protocols secure against strict polynomial-time adversaries, but all the results readily extend to the expected polynomial-time regime.

\subsection{Symmetric Functionalities Secure According to the Real/Ideal Paradigm}\label{sec:RealIdeal}
The model of secure computation we consider is defined in \cref{sec:RealIdeal:Def}, dominated functionalities are defined in \cref{sec:RealIdeal:dominated} and the impossibility results are stated and proved in \cref{sec:RealIdeal:LB}.

\subsubsection{Model Definition}\label{sec:RealIdeal:Def}

We provide the basic definitions for secure multiparty computation according to the real/ideal paradigm, for further details see \cite{Goldreich04}. Informally, a protocol is secure according to the real/ideal paradigm, if whatever an adversary can do in the real execution of protocol, can be done also in an ideal computation, in which an uncorrupted trusted party assists the computation. We consider \emph{full security}, meaning that the ideal-model adversary cannot prematurely abort the ideal computation.

\paragraph{Functionalities.}
\begin{definition}[functionalities]\label{def:func}
An $n$-party {\sf functionality} is a random process that maps vectors of $n$ inputs to vectors of $n$ outputs.\footnote{We assume that a functionality can be computed in polynomial time.} Given an $n$-party functionality $f \colon (\zs)^n \mapsto (\zs)^n$, let $f_i(\vx)$ denote its $i$'th output coordinate, \ie $f_i(\vx) = f(\vx)_i$. A functionality $f$ is {\sf symmetric}, if the output values of all parties are the same, \ie for every $\vx\in(\zs)^n$, $f_1(\vx)=f_2(\vx)=\ldots=f_n(\vx)$.
\end{definition}

\paragraph{Real-model execution.}
A real-model execution of an $n$-party protocol proceeds as described in \cref{sec:Protocols}.

\begin{definition} [real-model execution]\label{def:RealModel}
Let $\pi= (\Party_1,\ldots, \Party_n)$ be an $n$-party protocol and let $\CorSet \subseteq [n]$ denote the set of indices of the parties corrupted by $\Adv$. The {\sf joint execution of $\pi$ under $(\Adv,\CorSet)$ in the real model}, on input vector $\vx= (x_1,\ldots, x_n)$, auxiliary input $\aux$ and security parameter $\secParam$, denoted $\REAL_{\pi,\CorSet,\Adv(\aux)}(\vect{x},\secParam)$, is defined as the output vector of $\Party_1,\ldots,\Party_n$ and $\Adv(\aux)$ resulting from the protocol interaction, where for every $i \in \CorSet$, party $\Party_i$ computes its messages according to $\Adv$, and for every $j \notin \CorSet$, party $\Party_j$ computes its messages according to $\pi$.
\end{definition}

\paragraph{Ideal-model execution.}
An ideal computation of an $n$-party functionality $f$ on input $\vx=(x_1,\ldots,x_n)$ for parties $(\Party_1,\ldots,\Party_n)$ in the presence of an ideal-model adversary $\Adv$ controlling the parties indexed by $\CorSet\subseteq[n]$, proceeds via the following steps.
\begin{itemize}
  \item[Sending inputs to trusted party:]
  An honest party $\Party_i$ sends its input $x_i$ to the trusted party.
  The adversary may send to the trusted party arbitrary inputs for the corrupted parties. Let $x_i'$ be the value actually sent as the input of party $\Party_i$.

  \item[Trusted party answers the parties:]
  If $x_i'$ is outside of the domain for $\Party_i$, for some index $i$, or if no input was sent for $\Party_i$, then the trusted party sets $x_i'$ to be some predetermined default value.
  Next, the trusted party computes $f(x_1', \ldots, x_n')=(y_1, \ldots, y_n)$ and sends $y_i$ to party $\Party_i$ for every $i$.
  
  \item[Outputs:]
  Honest parties always output the message received from the trusted party and the corrupted parties output nothing.
  The adversary $\Adv$ outputs an arbitrary function of the initial inputs $\set{x_i}_{i\in\CorSet}$, the messages received by the corrupted parties from the trusted party $\set{y_i}_{i\in\CorSet}$ and its auxiliary input.
\end{itemize}

\begin{definition}[ideal-model computation]\label{def:IdealModel}
Let $f\colon(\zs)^n \mapsto (\zs)^n$ be an $n$-party functionality and let $\CorSet\subseteq [n]$. The {\sf joint execution of $f$ under $(\Adv, I)$ in the ideal model}, on input vector $\vect{x}=(x_1, \ldots, x_n)$, auxiliary input $\aux$ to $\Adv$ and security parameter $\secParam$, denoted $\IDEAL_{f,\CorSet,\Adv(\aux)}(\vect{x},\secParam)$, is defined as the output vector of $\Party_1, \ldots, \Party_n$ and $\Adv(\aux)$ resulting from the above described ideal process.
\end{definition}

\paragraph{Security definition.}

Having defined the real and ideal models, we can now define security of protocols according to the real/ideal paradigm.
\begin{definition}\label{def:SecureProtocol}
Let $f\colon(\zs)^n \mapsto (\zs)^n$ be an $n$-party functionality, and let $\pi$ be a probabilistic polynomial-time protocol computing $f$. The {\sf protocol $\pi$ $t$-securely computes $f$ (with computational security)}, if for every non-uniform polynomial-time real-model adversary \Adv, there exists a non-uniform (expected) polynomial-time adversary $\Sim$ for the ideal model, such that for every $\CorSet\subseteq [n]$ of size at most $t$, it holds that
$$
\set{\bigbrack \REAL_{\pi, \CorSet, \Adv(\aux)}(\vect{x}, \secParam)}_{(\vect{x}, \aux)\in(\zs)^{n+1}, \secParam\in\N}
\ci
\set{\bigbrack \IDEAL_{f, \CorSet, \Sim(\aux)}(\vect{x}, \secParam)}_{(\vect{x}, \aux)\in(\zs)^{n+1}, \secParam\in\N}.
$$
The {\sf protocol $\pi$ $t$-securely computes $f$ (with information-theoretic security)}, if for every real-model adversary \Adv, there exists an adversary $\Sim$ for the ideal model, whose running time is polynomial in the running time of $\Adv$, such that for every $\CorSet\subseteq [n]$ of size at most $t$,
$$
\set{\bigbrack \REAL_{\pi, \CorSet, \Adv(\aux)}(\vect{x}, \secParam)}_{(\vect{x}, \aux)\in(\zs)^{n+1}, \secParam\in\N}
\statclose
\set{\bigbrack \IDEAL_{f, \CorSet, \Sim(\aux)}(\vect{x}, \secParam)}_{(\vect{x}, \aux)\in(\zs)^{n+1}, \secParam\in\N}.
$$
\end{definition}

\subsubsection{Dominated Functionalities}\label{sec:RealIdeal:dominated}
A special class of symmetric functionalities are those with the property that every subset of a certain size can fully determine the output. For example, the multiparty Boolean AND and OR functionalities both have the property that every individual party can determine the output (for the AND functionality any party can always force the output to be 0, and for the OR functionality any party can always force the output to be 1). We distinguish between the case where there exists a single value for which every large enough subset can force the output and the case where different subsets can force the output to be different values.

\begin{definition}[dominated functionalities]\label{def:dominated}
A symmetric $n$-party functionality $f$ is {\sf weakly $k$-dominated}, if for every $k$-size subset $\CorSet\subseteq[n]$ there exists a polynomial-time computable value $\ys_\CorSet$, for which there exist inputs $\set{x_i}_{i\in\CorSet}$, such that $f(x_1,\ldots,x_n) = \ys_\CorSet$ for {\sf any} complementing subset of inputs $\set{x_j}_{j\notin\CorSet}$.
The functionality $f$ is {\sf $k$-dominated}, if there exists a polynomial-time computable value $\ys$ such that for every $k$-size subset $\CorSet\subseteq[n]$ there exist inputs $\set{x_i}_{i\in\CorSet}$, for which $f(x_1,\ldots,x_n) = \ys$ for {\sf any} subset of inputs $\set{x_j}_{j\notin\CorSet}$.
\end{definition}

\begin{example}
The function $f(x_1,x_2,x_3,x_4)=(x_1\wedge x_2) \vee (x_3\wedge x_4)$ is an example of a $4$-party function that is weakly $2$-dominated but not $2$-dominated. Every pair of input variables can be set to determine the output value. However, there is no single output value that can be determined by all pairs, for example, $\set{x_1,x_2}$ can force the output to be $1$ (by setting $x_1=x_2=1$) whereas $\set{x_1,x_3}$ can force the output to be $0$ (by setting $x_1=x_3=0$). The function
$$
f_{2\mhyphen{\rm of}\mhyphen 4}(x_1,x_2,x_3,x_4)=(x_1\wedge x_2) \vee (x_1\wedge x_3) \vee (x_1\wedge x_4) \vee (x_2\wedge x_3) \vee (x_2\wedge x_4) \vee (x_3\wedge x_4)
$$
is $2$-dominated with value $\ys=1$.
\end{example}

\begin{claim}\label{claim:dominated}
Let $f$ be an $n$-party functionality and let $m\leq\ThirdP$. If $f$ is weakly $m$-dominated, then it is $m$-dominated.
\end{claim}
\begin{proof}
Let $\CorSet_1,\CorSet_2\subseteq[n]$ be two subsets of size $m$. In case $\CorSet_1$ and $\CorSet_2$ are disjoint, consider the corresponding sets of input variables $\set{x_i}_{i\in\CorSet_1}$ and $\set{x_i}_{i\in\CorSet_2}$, and fix an arbitrary complementing subset of inputs $\set{x_j}_{j\notin\CorSet_1\cup\CorSet_2}$. On the one hand it holds that $f(x_1,\ldots,x_n)=\ys_{\CorSet_1}$ and on the other hand it holds that $f(x_1,\ldots,x_n)=\ys_{\CorSet_2}$, hence $\ys_{\CorSet_1}=\ys_{\CorSet_2}$.

In case $\CorSet_1$ and $\CorSet_2$ are not disjoint, it holds that $\size{\CorSet_1\cup\CorSet_2}<2m\leq\frac{2n}3$ and since $m\leq\ThirdP$, there exists a subset $\CorSet_3\subseteq [n]\setminus (\CorSet_1\cup\CorSet_2)$ of size $m$. Denote by $\ys_{\CorSet_3}$ the output value that can be determined by the input variables $\set{x_i}_{i\in\CorSet_3}$ ($\ys_{\CorSet_3}$ is guaranteed to exist since $f$ is weakly $m$-dominated). $\CorSet_3$ is disjoint from $\CorSet_1$ and from $\CorSet_2$, so it follows that $\ys_{\CorSet_1}=\ys_{\CorSet_3}$ and $\ys_{\CorSet_2}=\ys_{\CorSet_3}$, therefore $\ys_{\CorSet_1}=\ys_{\CorSet_2}$.
\end{proof}

\subsubsection{The Lower Bound}\label{sec:RealIdeal:LB}

\begin{lemma}\label{lem:MPCImp}
Let $n\geq 3$, let $t\geq \ThirdP$ and let $f$ be a symmetric $n$-party functionality that can be $t$-securely computed in the secure-channels \ptp model.
\begin{enumerate}
  \item
  If $\ThirdP\leq t<\HalfP$, then $f$ is $(n-2t)$-dominated.
  \item
  If $\HalfP\leq t<n$, then $f$ is $1$-dominated.
\end{enumerate}
\end{lemma}
\begin{proof}
Assume that $\ThirdP\leq t<\HalfP$ (the proof for $\HalfP\leq t<n$ is similar). Let $\pi$ be a protocol that $t$-securely computes $f$ in the \ptp model with secure channels. Since $f$ is symmetric, all honest parties output the same value (except for a negligible probability), hence $\pi$ is $(1-\negl,t)$-consistent; let $\Adv$ be the \ppt adversary guaranteed from \cref{lem:AdvStrictPoly} and let $\CorSet\subseteq [n]$ be any subset of size $n-2t$. It follows that given control over $\set{\Party_i}_{i\in\CorSet}$, \Adv can first fix a value $\ys_\CorSet$, and later force the output of the honest parties to be $\ys_\CorSet$ (except for a negligible probability). Since $\pi$ $t$-securely computes $f$ and $n-2t\leq t$, there exists an ideal-model adversary $\Sim$ that upon corrupting $\set{\Party_i}_{i\in\CorSet}$, can force the output of the honest parties in the ideal-model computation to be $\ys_\CorSet$. All $\Sim$ can do is to select the input values of the corrupted parties, hence, there must exist input values $\set{x_i}_{i\in\CorSet}$ that determine the output of the honest parties to be $\ys_\CorSet$, \ie $f$ is weakly $(n-2t)$-dominated. Since $n-2t\leq \ThirdP$ and following \cref{claim:dominated} we conclude that $f$ is $(n-2t)$-dominated.
\end{proof}

\subsection{Coin-Flipping Protocols}\label{sec:CF}
A coin-flipping protocol \cite{Blum81} allows the honest parties to jointly flip an unbiased coin, where even a coalition of cheating (efficient) parties cannot bias the outcome of the protocol by much. Our focus is on coin flipping, where the honest parties \emph{must} output a bit. Although \cref{lem:MPCImp} immediately shows that coin flipping cannot be securely computed according to the real/ideal paradigm, we present a stronger impossibility result by considering weaker security requirements.

\begin{definition}\label{def:CF}
A polynomial-time $n$-party protocol $\pi$ is a {\sf $(\gamma,t)$-bias} coin-flipping protocol, if the following holds.
\begin{enumerate}
  \item
  $\pi$ is $(1,t)$-consistent against \ppt adversaries.\footnote{Our negative result readily extends to protocols where consistency is only guaranteed to hold with high probability.}
  \item
  When interacting on security parameter $\secParam$ (for sufficiently large $\secParam$'s) with a \ppt adversary controlling at most $t$ corrupted parties, the common output of the honest parties is $\gamma(\secParam)$-close to the being a uniform bit.\footnote{In particular, the honest parties are allowed to output $\bot$, or values other than $\zo$, with probability at most $\gamma$.}
\end{enumerate}
\end{definition}

The following is a straightforward application of \cref{lem:AdvStrictPoly}.
\begin{lemma}\label{lem:ImpForCF}
In the secure-channels \ptp model, for $n\geq 3$ and $\gamma(\secParam) < \frac12 - 2^{-\secParam}$, there exists no $n$-party, $(\gamma,\ceil{\ThirdP})$-bias coin-flipping protocol.
\end{lemma}
\begin{proof}
Let $\pi$ be a \ptp $n$-party $(\gamma,\ceil{\ThirdP})$-bias coin-flipping protocol. Let \Adv be the \ppt adversary that is guaranteed by \cref{lem:AdvStrictPoly} (since $\pi$ is $(1,\ceil{\ThirdP})$-consistent against \ppt adversaries). Consider some fixed set of $\ceil{\ThirdP}$ corrupted parties of $\pi$ and let $Y(\secParam)$ denote the random variable of $\Adv(\secParam)$'s output in the first step of the attack. Without loss of generality, for infinitely many values of $\secParam$ it holds that $\pr{Y(\secParam) = 0} \le \frac12$. Consider the adversary $\Adv'$ that on security parameter $\secParam$, repeats the first step of $\Adv(\secParam)$ until the resulting value of $\ys$ is non-zero or $\secParam$ failed attempts have been reached, where if the latter happens $\Adv'$ aborts. Next, $\Adv'$ continues the non-zero execution of \Adv to make the honest parties of $\pi$ output $\ys$. It is immediate that for infinitely many values of $\secParam$, the common output of the honest parties under the above attack is $0$ with probability at most $2^{-\secParam}$, and hence the common output of the honest parties is $\frac12 - 2^{-\secParam}$ far from uniform. Thus, $\pi$ is not a $(\gamma,\ceil{\ThirdP})$-bias coin-flipping protocol.
\end{proof}

\section{Characterizing Secure Computation Without Broadcast}\label{sec:Characterization}
In this section we show that the lower bounds presented in \cref{lem:MPCImp} is tight. We treat separately the case where an honest majority is assumed and the case where no honest majority is assumed.

\subsection{No Honest Majority}\label{sec:Characterization:NoHonestMajority}
\citet[Thm.\ 7]{CohenLindell14} showed that, assuming the existence of one-way functions, any $1$-dominated functionality that can be $t$-securely computed in the broadcast model with authenticated channels, can also be $t$-securely computed in the \ptp model with authenticated channels.\footnote{The result in \cite{CohenLindell14} is based on the computationally-secure protocol in \cite[Thm.\ 2]{FGHHS02}. In the authenticated-channels \ptp model, this protocol requires one-way functions for constructing a consistent public-key infrastructure between the parties, to be used for authenticated broadcast.} Combining with \cref{lem:MPCImp}, we establish the following result.

\begin{theorem}[restating second part of \cref{thm:mainCharacterizationInf}]\label{thm:Characterization:NoHonestMajority}
Let $n\geq 3$, let $\HalfP\leq t<n$ and assume that one-way functions exist. An $n$-party functionality can be $t$-securely computed in the authenticated-channels \ptp model, if and only if it is $1$-dominated and can be $t$-securely computed in the authenticated-channels broadcast model.
\end{theorem}
\begin{proof}
Immediately by \cref{lem:MPCImp} and \citet[Thm.\ 7]{CohenLindell14}.
\end{proof}

\subsection{Honest Majority}\label{sec:Characterization:HonestMajority}

\begin{proposition}\label{prop:upperBound:honsetMajority}
Let $n\geq 3$, let $\ThirdP\leq t<\HalfP$, and let $f$ be a symmetric $n$-party functionality. If $f$ is $(n-2t)$-dominated, then it can be $t$-securely computed in the secure-channels \ptp model with information-theoretic security.
\end{proposition}

To prove \cref{prop:upperBound:honsetMajority} we use the \emph{two-threshold multiparty protocol} of \citet[Thm.\ 6]{FHHW03}. This protocol with parameters $t_1, t_2$ runs in the \ptp model with secure channels, and whenever $t_1\leq t_2$ and $t_1+2t_2<n$, the following holds. Let $\CorSet$ be the set of parties that the (computationally unbounded) adversary corrupts. If $\size{\CorSet}\le t_1$, then the protocol computes $f$ with full security. If $t_1< \size{\CorSet}\le t_2$, then the protocol securely computes $f$ with fairness (i.e., the adversary may force \emph{all} honest parties to output $\bot$, provided that it learns no new information).
In \cref{sec:thresholdSecurity}, we formally define the notion of two-threshold security.
This notion captures the security achieved by the protocol of \citet[Thm.\ 6]{FHHW03}.
\begin{theorem}[{\cite[Thm.\ 6]{FHHW03}}]\label{thm:6-FHHW03}
Let $n\geq 3$, let $t_1, t_2$ be parameters such that $t_1\leq t_2$ and $t_1+2t_2<n$, and let $f$ be an $n$-party functionality. Then, $f$ can be $(t_1,t_2)$-securely computed in the secure-channels \ptp model with information-theoretic security.
\end{theorem}

We now proceed to the proof of \cref{prop:upperBound:honsetMajority}.
\begin{proof}[Proof of \cref{prop:upperBound:honsetMajority}]
Let $f$ be an $(n-2t)$-dominated functionality with default output value $\ys$. If $n-2t = 1$, then $f$ is $1$-dominated, and since $t<\HalfP$, $f$ can be $t$-securely computed with information-theoretic security in the secure-channels broadcast model (\eg using \citet{RB89}). Hence, the proposition follows from \cite[Thm.\ 7]{CohenLindell14}.\footnote{When an honest majority is assumed, the result in \cite{CohenLindell14} can be adjusted to use the information-theoretically secure protocol in \cite[Thm.\ 3]{FGHHS02}. In the secure-channels \ptp model, this protocol uses information-theoretically pseudo-signatures \cite{PW92} for computing a setup, to be used for authenticated broadcast.}

For $n-2t \ge 2$, set $t_1=n-2t-1$ and $t_2=t$, and let $\pi'$ be the $n$-party protocol, guaranteed to exist by \cref{thm:6-FHHW03}, that $(t_1,t_2)$-securely computes $f$. We define $\pi$ to be the following $n$-party protocol for computing $f$ in the \ptp model with secure channels.
\begin{protocol}[$\pi$]~
\begin{enumerate}
  \item \label{step:1-pi}
  The parties run the protocol $\pi'$.
  Let $y_i$ be the output of $\Party_i$ at the end of the execution.
	
  \item \label{step:2-pi}
  If $y_i\neq\bot$, party $\Party_i$ outputs $y_i$, otherwise it outputs $\ys$.
\end{enumerate}
\end{protocol}

Let $\Adv$ be an adversary attacking the execution of $\pi$ and let $\CorSet\subseteq[n]$ be a subset of size at most $t$. It follows from \cref{thm:6-FHHW03} that there exists a (possibly aborting) adversary $\Sim'$ for \Adv in the $t_1$-threshold ideal model such that
$$
\set{\bigbrack \REAL_{\pi', \CorSet, \Adv(\aux)}(\vect{x}, \secParam)}_{(\vect{x}, \aux)\in(\zs)^{n+1}, \secParam\in\N}
\statclose
\set{\bigbrack \IDEAL^{t_1}_{f, \CorSet, \Sim'(\aux)}(\vect{x}, \secParam)}_{(\vect{x}, \aux)\in(\zs)^{n+1}, \secParam\in\N}.
$$

Using $\Sim'$, we construct the following non-aborting adversary $\Sim$ for the full-security ideal model. On inputs $\set{x_i}_{i\in\CorSet}$ and auxiliary input $z$, $\Sim$ starts by emulating $\Sim'$ on these inputs, playing the role of the trusted party (in the $t_1$-threshold ideal model). If $\Sim'$ sends an \abort command, it is guaranteed that $\size{\CorSet}\geq n-2t$ and since $f$ is $(n-2t)$-dominated, there exist input values $\set{x'_i}_{i\in\CorSet}$ that determine the output of $f$ to be $\ys$. So in this case, $\Sim$ sends these $\set{x'_i}_{i\in\CorSet}$ to the trusted party (in the full-security ideal model) and returns $\bot$ to $\Sim'$. Otherwise, $\Sim'$ does not abort and $\Sim$ forwards the message from $\Sim'$ to the trusted party and the answer from the trusted party back to $\Sim'$. In both cases $\Sim$ outputs whatever $\Sim'$ outputs and halts.

A main observation is that the views of the adversary $\Adv$ in an execution of $\pi$ and in an execution of $\pi'$ (with the same inputs and random coins) are identical. This holds since the only difference between $\pi$ and $\pi'$ is in the second step of $\pi$ that does not involve any interaction. It follows that in case the output of the parties in Step~\ref{step:1-pi} of $\pi$ is not $\bot$, the joint distribution of the honest parties' output and the output of $\Adv$ in $\pi$ is statistically close to the output of the honest parties and of $\Sim$ in the full-security ideal model (since the later is exactly the output of the honest parties and of $\Sim'$ in the $t_1$-threshold ideal model). If the output in Step~\ref{step:1-pi} of $\pi$ is $\bot$, then all honest parties in $\pi$ output $\ys$. In this case $\Sim'$ sends \abort (except for a negligible probability) and since $\Sim$ sends to the trusted party the input values $\set{x'_i}_{i\in\CorSet}$ that determine the output of $f$ to be $\ys$, the honest parties' output is $\ys$ also in the ideal computation. We conclude that
$$
\set{\bigbrack \REAL_{\pi, \CorSet, \Adv(\aux)}(\vect{x}, \secParam)}_{(\vect{x}, \aux)\in(\zs)^{n+1}, \secParam\in\N}
\statclose
\set{\bigbrack \IDEAL_{f, \CorSet, \Sim(\aux)}(\vect{x}, \secParam)}_{(\vect{x}, \aux)\in(\zs)^{n+1}, \secParam\in\N}.
$$
\end{proof}

\begin{theorem}[restating the first part of \cref{thm:mainCharacterizationInf}]\label{thm:Characterization:HonestMajority}
Let $n\geq 3$ and $\ThirdP\leq t<\HalfP$. A symmetric $n$-party functionality can be $t$-securely computed in the secure-channels \ptp model, if and only if it is $(n-2t)$-dominated.
\end{theorem}
\begin{proof}
Immediately follows by \cref{lem:MPCImp} and \cref{prop:upperBound:honsetMajority}.
\end{proof}


\subsubsection{Defining Two-Threshold Security}\label{sec:thresholdSecurity}
We present a weaker variant of the ideal model that allows for a premature (and fair) abort, in case sufficiently many parties are corrupted. Next, we define two-threshold security of protocols.
\paragraph{Threshold ideal-model execution.}
A $t$-threshold ideal computation of an $n$-party functionality $f$ on input $\vx=(x_1,\ldots,x_n)$ for parties $(\Party_1,\ldots,\Party_n)$, in the presence of an ideal-model adversary $\Adv$ controlling the parties indexed by $\CorSet\subseteq[n]$, proceeds via the following steps.
\begin{itemize}
  \item[Sending inputs to trusted party:]
  An honest party $\Party_i$ sends its input $x_i$ to the trusted party.
  The adversary may send to the trusted party arbitrary inputs for the corrupted parties. If $\size{\CorSet}>t$, then the adversary may send a special \abort command to the trusted party. Let $x_i'$ be the value actually sent as the input of party $\Party_i$.

  \item[Trusted party answers the parties:]
  If the adversary sends the special \abort command (specifically, $\size{\CorSet}>t$), then the trusted party sends $\bot$ to all the parties. Otherwise, if $x_i'$ is outside of the domain for $\Party_i$, for some index $i$, or if no input is sent for $\Party_i$, then the trusted party sets $x_i'$ to be some predetermined default value. Next, the trusted party computes $f(x_1', \ldots, x_n')=(y_1, \ldots, y_n)$ and sends $y_i$ to party $\Party_i$ for every $i$.

  \item[Outputs:]
  Honest parties always output the message received from the trusted party and the corrupted parties output nothing.
  The adversary $\Adv$ outputs an arbitrary function of the initial inputs $\set{x_i}_{i\in\CorSet}$, the messages received by the corrupted parties from the trusted party $\set{y_i}_{i\in\CorSet}$ and its auxiliary input.
\end{itemize}

\begin{definition}[Threshold ideal-model computation]\label{def:IdealMode_threshold}
Let $f\colon(\zs)^n \mapsto (\zs)^n$ be an $n$-party functionality and let $\CorSet\subseteq [n]$. The {\sf joint execution of $f$ under $(\Adv, I)$ in the $t$-threshold ideal model}, on input vector $\vect{x}=(x_1, \ldots, x_n)$, auxiliary input $\aux$ to \Adv and security parameter $\secParam$, denoted $\IDEAL^t_{f,\CorSet,\Adv(\aux)}(\vect{x},\secParam)$, is defined as the output vector of $\Party_1, \ldots, \Party_n$ and $\Adv(\aux)$ resulting from the above described ideal process.
\end{definition}

\begin{definition}\label{def:SecureProtocol_threshold}
Let $f\colon(\zs)^n \mapsto (\zs)^n$ be an $n$-party functionality, and let $\pi$ be a probabilistic polynomial-time protocol computing $f$. The {\sf protocol $\pi$ $(t_1,t_2)$-securely computes $f$ (with information-theoretic security)}, if for every real-model adversary \Adv, there exists an adversary $\Sim$ for the $t_1$-threshold ideal model, whose running time is polynomial in the running time of $\Adv$, such that for every $\CorSet\subseteq [n]$ of size at most $t_2$
$$
\set{\bigbrack \REAL_{\pi, \CorSet, \Adv(\aux)}(\vect{x}, \secParam)}_{(\vect{x}, \aux)\in(\zs)^{n+1}, \secParam\in\N}
\statclose
\set{\bigbrack \IDEAL^{t_1}_{f, \CorSet, \Sim(\aux)}(\vect{x}, \secParam)}_{(\vect{x}, \aux)\in(\zs)^{n+1}, \secParam\in\N}.
$$
\end{definition}

\bibliographystyle{abbrvnat}
\bibliography{crypto}

\end{document}